\documentclass[a4paper,USenglish]{lipics-v2021}

\pdfoutput=1 
\hideLIPIcs  

\nolinenumbers

\usepackage[T1]{fontenc}
\usepackage{framed}
\usepackage{hyperref}
\usepackage{subcaption}
\usepackage{float}
\usepackage[linesnumbered]{algorithm2e}
\usepackage{comment}
\usepackage{pgfplots}
\pgfplotsset{compat=1.18}
\usepackage{tikz}
\usepackage{environ}
\makeatletter
\newsavebox{\measure@tikzpicture}
\NewEnviron{scaletikzpicturetowidth}[1]{%
	\def\tikz@width{#1}%
	\def\tikzscale{1}\begin{lrbox}{\measure@tikzpicture}%
		\BODY
	\end{lrbox}%
	\pgfmathparse{#1/\wd\measure@tikzpicture}%
	\edef\tikzscale{\pgfmathresult}%
	\BODY
}
\makeatother
\usepackage{xintexpr}
\usepackage[normalem]{ulem}
\newcommand{\kaoutar}[1]{{\color{magenta}{\bf kaoutar:} #1}}

\theoremstyle{definition}
\newtheorem{definition2}{Definition}

\theoremstyle{definition}
\newtheorem{property}{Property}

\newcommand{\rememberlines}{\xdef\rememberedlines{\number\value{AlgoLine}}}
\newcommand{\resumenumbering}{\setcounter{AlgoLine}{\rememberedlines}}

\usepackage{babel}

\title{Arma: Byzantine Fault Tolerant Consensus with Horizontal Scalability}

\author{Yacov Manevich}{IBM Research, Zurich}{yacov.manevich@ibm.com}{}{}
\authorrunning{ }

\author{Hagar Meir}{IBM Research, Haifa}{hagar.meir@ibm.com}{}{}

\author{Kaoutar Elkhiyaoui}{IBM Research, Zurich}{kao@zurich.ibm.com}{}{}

\author{Yoav Tock}{IBM Research, Haifa}{tock@il.ibm.com}{}{}

\author{May Buzaglo}{IBM Research, Haifa}{may.buzaglo@ibm.com}{}{}


\begin{document}
	
\maketitle

\begin{abstract}

Arma\footnote{Means ``Chariot'' in Greek. Similarly to a chariot that can tow
	more when horses are added, Arma totally orders more client transactions
	the more hardware is added.} is a Byzantine Fault Tolerant (BFT) consensus system designed to
achieve horizontal scalability across all hardware resources: network
bandwidth, CPU, and disk I/O. As opposed to preceding
BFT protocols, Arma separates the dissemination and validation of
client transactions from the consensus process, restricting the latter to totally ordering only metadata of batches of transactions. This separation enables
each party to distribute compute and storage resources for transaction validation, dissemination and disk I/O among multiple machines, resulting in horizontal scalability.
Additionally, Arma ensures censorship resistance by imposing a maximum
time limit on the inclusion of client transactions. We built and evaluated two Arma prototypes. The first is an independent system handling over 200,000 transactions per second, the second integrated into Hyperledger Fabric, speeding its consensus by an order of magnitude.
\end{abstract}

\section{Introduction}
\label{sec:intro}


Byzantine Fault Tolerant (BFT) consensus has gained attention with the advent of Distributed Ledger Technologies and the extended use of the latter for financial asset exchange applications, i.e. scenarios that desire the high degree of resilience that BFT consensus offers. 
However, early works on BFT protocols suffered from poor performance, especially with a high number of participants. 
In recent years, academic and industrial research successfully demonstrated that BFT systems are capable of achieving significant throughput at large scale.


A prevalent characteristic among consensus protocols such as PBFT \cite{PBFT} or Paxos \cite{Lamport2001PaxosMS} is their monolithic architecture, where each participating party consolidates all the subroutines of the protocol within a single node.
The reason for such an approach is twofold: First, implementing and analyzing a distributed consensus protocol involving multiple parties is inherently challenging, and having each party manage multiple nodes further exacerbates the complexity. 
Second, many researchers gauge protocol efficiency based on the overall number of messages sent during the protocol or by each party, often overlooking the significant discrepancy in message sizes and their transmission efficiency.
Indeed, the network bandwidth a protocol requires from each node is often the limiting factor for the throughput and not the total number of messages sent in the protocol.

While the monolithic architecture is relatively straightforward to implement and analyze, it inherently constraints the performance of the system. 
Eventually, the machine of each party will reach its limit in terms of CPU and storage operations per second.
Only recently, researchers have proposed a consensus protocol \cite{Narwhale} that allows each party to distribute its execution across multiple machines, where not all nodes are created equal. 
This design enables the horizontal scaling of all resources required for consensus, including CPU, storage I/O rate, and network bandwidth. 
In such systems, each party operates several nodes, and the various subroutines of the consensus protocol are executed on nodes based on their specific roles. 
However, current implementations of such systems do not enable censorship resistance nor transaction de-duplication.

\subsection{The Contribution of This Work}
This paper introduces Arma, a new consensus protocol that builds on ideas from the recent work of Danezis et al.~\cite{Narwhale}, to achieve horizontal scalability. 
Distinguishing itself from the previous work, Arma incorporates censorship resistance and transaction de-duplication mechanisms, where in~\cite{Narwhale} these properties are not inherent like in Arma, and can only be achieved with having the transaction totally ordered several times, incurring a severe impact on performance. 
%
%
Once a client submits a transaction to Arma, it can be assured of its finalization.
As a result, Arma not only demonstrates high throughput but also enhances the user experience, by eliminating the need for client applications to include logic for tracking transaction finalization. 
%
%
This advancement opens up possibilities for constructing payment systems where clients can effortlessly initiate transactions without the necessity of monitoring their execution.

Additionally, we highlight a new technique of denial of service protection that utilizes probabilistic verification of transaction batches which allows Arma to be resistant to denial of service at a performance overhead considerably lower than the state of the art. 

\subsection{Brief Overview of the Arma Protocol}

In the Arma consensus protocol, transaction dissemination and validation are entrusted to separate groups of nodes known as shards. 
Each shard includes a representative from every party (participant or member who runs the Arma protocol) and exclusively manages a specific subset of transactions.
The nodes within these shards provide attestations of persisted batches of transactions or votes related to node misbehavior. 
These attestations and votes are then submitted into a BFT consensus protocol. By totally ordering these attestations and votes, Arma nodes gain collective knowledge regarding which batches of transactions have been safely persisted by a sufficient number of nodes and identifies nodes displaying faulty behavior, such as deliberate transaction censorship or a simple crash.
With the attestations and votes fully ordered through the BFT consensus protocol, an ordering of batches from different shards is achieved, subsequently leading to an implicit total order between all transactions.

\subsection{Paper Outline}

Firstly, Section~\ref{sec:related} surveys the existing work related to this paper and provides a comparison.
Then, Section~\ref{sec:design} presents an overview of the design of Arma, highlighting its key components and mechanisms, while Section~\ref{sec:arch} gives an in depth description of the Arma architecture. 
Following that, Section~\ref{sec:integration} explains how to integrate Arma into Hyperledger Fabric~\cite{HLF}, and discusses performance evaluation in Section \ref{sec:eval}.
In Section~\ref{sec:proof} we give a formal analysis of Arma's properties, and in Section~\ref{sec:future} we conclude and discuss future work.

\section{Related Work}
\label{sec:related}
In this section, we provide an overview of existing BFT protocols. 
We distinguish between earlier BFT architectures and the one Arma is based on.

In contrast to permissionless blockchains like Bitcoin~\cite{bitcoin} and Ethereum~\cite{ethereum}, systems designed for enterprise use cases often possess a well-defined membership structure and impose stringent performance requirements. 
Among the various BFT protocols, PBFT \cite{PBFT} stands out as a prominent choice for such settings, serving as the foundation upon which subsequent works build.
A key characteristic shared by these protocols is the establishment of a total order by having a designated leader node broadcast a batch of client transactions. 
However, this inherent asymmetry among the nodes limits the protocol's throughput, as it becomes constrained by the network bandwidth of the leader node. 
Indeed, any protocol that relies on a sole leader to propose batches and coordinate the consensus rounds, imposes a high load on the leader, which unavoidably becomes a bottleneck.
A significant approach to mitigate the aforementioned bottleneck is to have multiple leaders (e.g.,~\cite{Red-Belly, Honey, Bounded-Delay, Mir-BFT}), many nodes who concurrently and independently propose batches.

A notable example is Mir-BFT~\cite{Mir-BFT}, which uses multiple leaders while employing a transaction de-duplication technique.
Mir-BFT runs multiple parallel instances of PBFT~\cite{PBFT}, where each batch sequence $i$ is specifically reserved for the PBFT instance $i\,\%\,k$, where $k$ is the number of parallel PBFT instances. 
By running multiple leaders, Mir-BFT effectively distributes
the network load that would typically be concentrated on a single
leader. 
This distribution helps alleviate the network bandwidth bottleneck associated with having a single leader.
Furthermore, in Mir-BFT every client transaction is assigned to a specific PBFT instance by hashing it to a value within the range of $\left\{ 1,...,k\right\}$.
The deterministic hashing of transactions to PBFT instances
ensures that duplicate transactions are prevented. This feature is
crucial not only for preventing denial-of-service attacks but also
for avoiding unnecessary duplication of total order for the same
transaction.


%

However, running BFT instances in parallel can indeed increase throughput, but it also introduces a higher vulnerability to node crashes. 
Let's consider a protocol like Mir-BFT, where there are $k=\frac{n}{2}$ parallel instances of BFT, where $n$ is the total number of nodes.
In this scenario, batch $b_{i}$ is assembled and broadcast by the leader of instance $i\,\%\,k$ , and similarly, batch $b_{i+1}$ is broadcast by the leader of instance $(i+1)\,\%\,k$. 
If the leader of instance $i\,\%\,k$ crashes before broadcasting batch $b_{i}$, batch $b_{i+1}$ cannot be delivered until a new leader is established in instance $i\,\%\,k$. 
As a result, in such a protocol with k parallel PBFT instances, the throughput drops to zero when any of the leaders among the $k$ different instances faults.
If the probability of a crash is the same across all nodes and they are independent, then the probability of a failure is higher in Mir-BFT like protocols than PBFT (with a single leader).
In a standard PBFT protocol, the throughput drops to zero only when the leader node crashes.
The increased susceptibility to crashes in Mir-BFT's approach can have a significant impact on the system's up-time availability.

Moreover, nodes participating in the consensus process not only  transmit batches over the network but also perform other essential tasks. 
These tasks include writing the batches to disk for crash fault tolerance, parsing
the client transactions within the batches, and verifying their integrity.
The speed at which batches are written to disk is influenced by the
underlying storage system of the node. On the other hand, transaction
verification is a computationally intensive task that relies on the
CPU's processing power. Both disk writing and transaction verification
play crucial roles in the overall performance of the consensus process.
Consequently, a protocol that is efficient at load balancing network bandwidth
across all parties may have its limiting factor shift towards CPU and disk I/O.

\subsection{Distributed Multi-Instance BFT Architecture}

Running PBFT in parallel enhances horizontal scalability, but it also introduces certain limitations. 
In this setup, each node participates in multiple PBFT protocols, resulting in the need to verify all transactions and write all batches to its underlying storage. 
Consequently, the scalability of Mir-BFT is primarily limited to the vertical scaling of CPU and storage I/O. 
Simply adding more machines to the system cannot effectively address potential CPU or storage I/O bottlenecks that may arise.


Danezis et al.~\cite{Narwhale} offer insights into spreading the load of storage I/O and CPU across multiple machines within the same party, thereby overcoming these limitations. 
The key concept is to separate and distribute the tasks of transaction dissemination and validation across multiple machines and subsequently establish a total order
for the corresponding metadata. 
By effectively separating these responsibilities and achieving a distributed consensus on the metadata, the aforementioned work \cite{Narwhale} enables the scalability of storage I/O and CPU by leveraging the collective resources of multiple machines within a party.

Figure \ref{fig:multi} depicts the architectural difference between multi-instance BFT protocols such as Mir-BFT and distributed multi-instance BFT like Arma and \cite{Narwhale}. 

\begin{figure}[htp]
	\centering
	\subfloat[\centering Each party (a,b,c,d) runs a single machine with co-located BFT instances]{{\includegraphics[width=5cm]{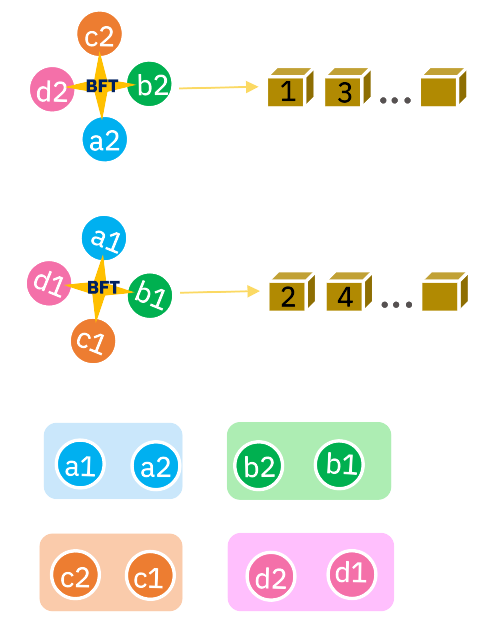} }}%
	\qquad
	\subfloat[\centering Each party (a,b,c,d) runs several machines, each executing a dedicated protocol instance]{{\includegraphics[width=7cm]{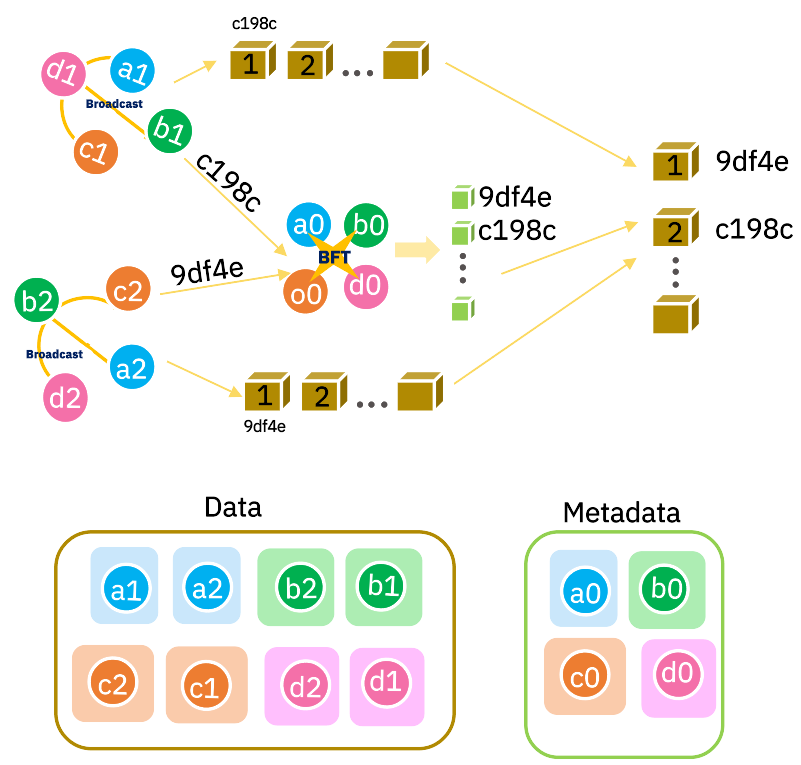} }}%
	\caption{(a) Multi-Instance BFT  vs (b) Distributed Multi-Instance BFT }%
	\label{fig:multi}
\end{figure}

In distributed multi-instance BFT the actual dissemination of transactions is not conducted through a full consensus protocol but rather through a sub-protocol known as reliable broadcast. 
In \cite{Narwhale}, once a batch is persisted by enough nodes, a certificate of availability is formed, including signatures on the metadata of the batch.
The certificate of availability is then subjected to a real consensus protocol to establish total order. 
The total order of transactions is derived from the order of the metadata of the batches.

Similar to Mir-BFT or other multi-instance BFT protocols, there can be multiple instances of the sub-protocol responsible for transaction dissemination. 
However, each instance can execute on its own machine, and instances operated by the same party are not co-located on the same machine. 
This distributed execution allows for scalability and better resource utilization across multiple machines, contributing to the overall efficiency and performance of the consensus protocol.
Furthermore, unlike the multi-instance BFT approach, a distributed multi-instance BFT design is less vulnerable to a crash of a party. 
In the event of a node responsible for broadcasting batches experiencing a crash, other instances can still completely order certificates of availability.
This resilience is attributed to the fact that the order among instances transmitting transactions is determined by the consensus protocol itself, rather than relying on a predetermined sequence allocation as seen in multi-instance BFT systems like Mir-BFT. 
This feature enhances the fault tolerance and availability of the system, as it allows for continuous progress even in the presence of node failures.

\subsubsection{Distributed Multi-Instance BFT Shortcomings}


In the work of 'Narwhal and Tusk' \cite{Narwhale}, once a batch creator gets enough distinct signed acknowledgments for the batch, it combines them into a certificate of batch availability and then sends it to the rest of the nodes. 
However, this design is susceptible to the batch creator maliciously withholding the certificate of batch availability, hence delaying the transactions in the batch from being finalized. 
The authors point that such a problem is resolved by having the clients monitor their transaction inclusion and if a timeout has been reached, re-submitting it to a different shard. 

In contrast to the aforementioned design, in Arma, each batcher node be it the batch creator (denoted as primary) or a secondary is responsible for having its signed acknowledgment to be totally ordered via consensus. 
Therefore, Arma is not susceptible to such an attack as it totally orders the signed acknowledgments in a decentralized manner.

Moreover, the approach  of \cite{Narwhale} does not  incorporate censorship
resistance or transaction de-duplication mechanisms. A malicious node
responsible for broadcasting batches can selectively ignore transactions
or intentionally slow down its execution, potentially leading to delays.
In such cases, the system relies on clients to resubmit their transactions
if they are not finalized within a designated timeout period. Furthermore,
malicious clients can submit their transactions to multiple instances,
resulting in wasteful resource utilization and potential duplication
of finalized transactions.

In contrast, Arma builds upon the fundamental idea of \cite{Narwhale} and thus benefits from its inherent performance advantages, but also incorporates
mechanisms to achieve censorship resistance (as defined in \cite{Honey}) and transaction de-duplication.
By following the Arma protocol, clients can be certain that their
transactions will always be ordered, eliminating the need for timeout-based
resubmission. 

\section{Design Overview}
\label{sec:design}

The Arma protocol is broken down into four stages: \textbf{Validation}, \textbf{Batching}, \textbf{Consensus}, and \textbf{Block Assembly}.
In the first stage, transactions are validated against some pre-defined static system rules (e.g., check transaction format, and client signature);
once the validation stage completes, transactions are dispatched to nodes responsible for batching, named \emph{batchers}. 
During batching, transactions are bundled into batches and persisted to disk.
The batchers then create \emph{batch attestation shares (BASs)}, which are batch digests and signatures over them, and submit the BASs to the \emph{consensus nodes}. 
The latter totally order the BASs and create signed block headers that are delivered to the block assembly nodes, which are called \emph{assemblers}. 
After receiving the block headers from the consensus the assemblers fetch the corresponding transaction batches from the batchers and construct the full block.

It is important to note that in Arma's validation stage which takes place at the routers, only a local transaction validation is performed to ascertain their origin and prevent denial of service attacks.  However, the global validity of the transactions is not established by Arma, as Arma does not execute nor validate any transactions in the batches it totally orders in parallel. 
Therefore, any system that wishes to incorporate Arma must retrieve the blocks from the assemblers and then execute and validate the transactions, and potentially,  some transactions would be marked as invalid, and skipped over.

%
%
%

\subsection{Terminology and Definitions}

We use the following terms to describe the Arma protocol:

\begin{definition2}[\textbf{Nodes}]
Each node has a dedicated role (router, batcher, consensus node or assembler), and it can be either a physical or a virtual machine, or a  process that shares the machine with other nodes but operates independently. 
\end{definition2}

\begin{definition2}[\textbf{Parties}]
The Arma consensus participants. Each node in the system is associated with a single party, that deploys and maintains the node, and each  
party comprises router, batcher, consenus and assembler nodes. 
We denote the total number of parties in the system as $N$. 
\end{definition2}

\begin{definition2}[\textbf{Byzantine node/party}]
A node or party that is not necessarily abiding by the protocol, but exhibits arbitrary behavior. Byzantine \cite{ByzGen} nodes/parties not only crash but they may deviate from the specified protocol. Byzantine nodes/parties typically represent entities that have been compromised. 

In what follows we assume that there is at most $F$ out of $N$ Byzantine parties and that $F < \frac{N}{3}$. 
We consider a node to be Byzantine if the party with which it associates is also Byzantine and vice versa.
\end{definition2}

\begin{definition2}[\textbf{Quorum}]
Defined by the smallest subset of distinct parties which is guaranteed to intersect with another quorum in at least one correct party.
For instance, if $N=3F+1$, a quorum is exactly $2F+1$ parties, as
$\left(2F+1\right)\cdot2-N=F+1$ which include at least one correct
party as up to $F$ are malicious.
\end{definition2}

\begin{definition2}[\textbf{Shard}]
In Arma, a shard represents a logical partition of the transaction space, dividing it into sets of comparable sizes. 
For instance, one possible approach is to assign the least significant bit of a transaction to Shard 1 if it is zero, and to Shard 2 if it is one. 
Sharding enables parallelization of transaction processing, and hence increases scalability.
\end{definition2}

\begin{definition2}[\textbf{Batch}]
A set of transactions from a specific shard bundled together to be sent over the network or persisted to disk.
\end{definition2}

\begin{definition2}[\textbf{Batch Attestation Share (BAS)}]
A message signed by a party attesting that a specific batch (from a specific shard) has been persisted to disk by that party. 
\end{definition2}

\begin{figure}[H]
		\centering
		\includegraphics[width=9cm]{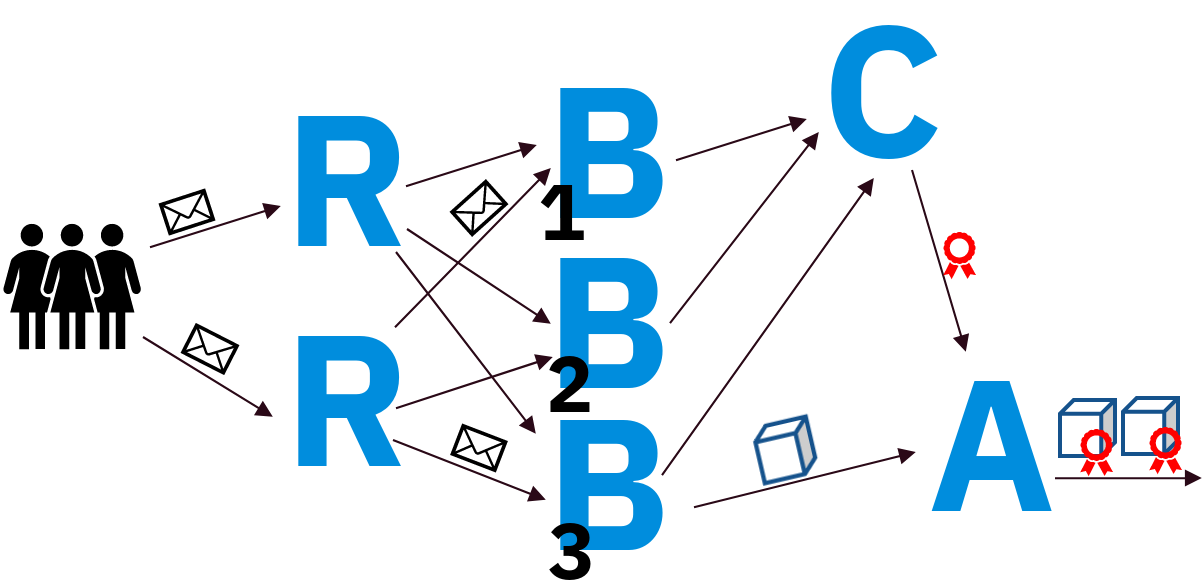}
		\caption{Information flow among components run by a single party.
			Each party runs several (R)outer nodes, (B)atcher nodes (where each batcher node is mapped to a shard), a single (C)onsensus node and at least one (A)ssembler node.
			(R)outer nodes dispatches transactions submitted by clients to the shards via their (B)atcher nodes. 
			The (B)atcher nodes send votes on persisted batches to the (C)onsensus nodes for ordering. 
			The (A)ssembler node receives from the (C)onsensus node ordered signatures over the digests of batches persisted by the (B)atcher nodes. The (A)ssembler node then outputs the blocks which consists of the ordered batches and the signatures over their digests. 
		}%
		\label{fig:party}
\end{figure}

\subsection{Node Roles }
In the Arma system, each party operates one or multiple nodes of different types: \emph{Router, Batcher, Consensus}, and \emph{Assembler}. 
In particular, a party's infrastructure would typically include one or more routers, multiple batchers (one per shard), a single consensus node, and at least one assembler.
Figure \ref{fig:party} depicts the components of a single party in an Arma system with three shards. 
The Router nodes check transaction validity and forward the received and validated transactions to the Batcher of the  designated shard.
Batchers group the received transactions into batches and submit BASs to the Consensus nodes after persisting the batches locally. 
Assemblers receive an ordered list of signed block headers from the Consensus nodes, and retrieve the corresponding batches from Batchers. 

More specifically: 

\begin{enumerate}
	\item \textbf{Router}: Clients submit transactions to the routers of multiple parties.
	The routers operate in a stateless manner and verify that the transactions are well formed and properly signed by the authorized clients. 
	Each router then deterministically assigns validated transactions to shards and forwards them to the corresponding batcher node that belongs to the same party. 
	
	\item \textbf{Batcher}: A batcher receives validated transactions from the router(s) and store them in its \emph{memory pool}. 
	For each shard, there is exactly one batcher from each party.
	A batcher can be either a \emph{primary batcher} or a \emph{secondary batcher}:
	\begin{itemize}
		\item A primary batcher persists submitted transactions to disk and bundles them into a batch. 
		It subsequently forwards the formed batches to the secondary batchers. 
		Each shard has exactly one primary batcher, for a given term.
		\item A secondary batcher pulls batches from the primary batcher and persists them to disk. 
	\end{itemize}
	Upon persistence of a transaction by a primary, the transaction is removed from its memory pool. 
	Similarly, once a transaction enters a secondary
	batcher node, it remains in its memory pool until it receives a batch containing this transaction from the primary. 
	Eventually batchers would submit BASs to the consensus nodes to be totally ordered.
	\item \textbf{Consensus}: The consensus nodes collectively agree on the total order of transactions by ordering the batches of various shards.
	They do it by running a BFT consensus protocol which eventually assigns each batch a unique sequence number. 
	The consensus nodes output a series of batch headers, each signed by a quorum of consensus nodes.
	\item \textbf{Assembler}: An assembler compiles the total order of transactions submitted to Arma, by combining the output of consensus nodes, i.e., the total order of batch headers, with batches it retrieves from the batchers to form blocks.
	The assemblers persist the blocks to disk and can be seen as archivists, as the blocks contain both information from the consensus nodes and the batchers.
	
\end{enumerate}

%

\section{Detailed Architecture}
\label{sec:arch}

Now that the stages of the Arma protocol are explained, and the various roles of the nodes are defined, we take a deep dive into each of the four stages of the protocol.

\subsection{Routing and Validation}
In Arma, a correct client sends her transaction to all parties.
More specifically, the transaction is submitted to the routers of each party.
A router node is the simplest component among the four, operating in
a stateless manner with two roles: 
\begin{itemize}
	\item Mapping transactions to shards and forwarding them to the corresponding batcher nodes of one's party.
	\item Performing validity checks on the transaction and dropping it if it is found invalid. 
\end{itemize}

Therefore, a router node is a stateless component and of which multiple instances can be added as needed, making hence transaction validation and routing horizontally scalable.

\subsubsection{Mapping Transactions to Shards}

To ensure censorship resistance and de-duplication, it is essential that transactions are mapped to shards in a deterministic manner.
In this way, assuming a client follows the protocol, and submits the transaction to all parties, all correct parties' routers would  forward the received transactions to the appropriate batchers, ensuring that all correct batchers of the shard will receive the
transaction.
Additionally, for effective load balancing across the shards, the router must distribute the transactions among the known $k$ shards in a manner that is as close to a uniform distribution across the set $\left\{ 1,..,k\right\} $ as possible.

While a cryptographic hash function could serve as a suitable choice, it is worth noting that the router does not require the collision resistance and one-way properties. Instead, more efficient alternatives like a CRC checksum are enough.

\subsubsection{Transaction Verification}
Besides forwarding transactions to their corresponding shards, routers also ensure that transactions are well formed and properly signed by clients. 
In earlier protocols such as \cite{Mir-BFT}, each party ran a single node verifying all transactions arriving from clients to that party.
However, in Arma, a party can utilize the routers to horizontally scale the CPU intensive task of verifying signatures.
Identifying invalid transactions and excluding them from the next steps, ensure that system resources are devoted only to the processing of correctly formed transactions.

\subsection{Transaction Batching} 
\label{sec:batching}
In the batching stage transactions are delivered to batchers and bundled together into batches, in a way that ensures that (i) transaction batches are persisted safely into stable storage for redundancy and retrieval, and (ii) the corresponding batch attestation shares are submitted to the consensus stage. 
Recall that batchers are grouped into shards.
Each party runs a single batcher for every shard in the system and that batcher can be either a primary or a secondary. 
Each shard has a single designated primary batcher node, while the rest are secondaries. 

\subsubsection{Transaction Dissemination}
As mentioned before, a transaction is submitted by a correct client to all the parties in the system (i.e., their router nodes). 
The correct client considers the submission of the transaction successful only if it has been successfully delivered to $N-F$ parties. 
The routers of the correct parties subsequently forward the transaction to its assigned shard via the shard's corresponding batcher for each party. 

If the transaction is delivered to the primary batcher of the shard, then the primary batcher includes it in a batch and persists the batch to disk.
When the transaction arrives to a secondary batcher, the batcher puts it into its memory pool. 
The secondary batcher then pulls batches from the primary and after persisting them to disk, removes transactions that appear in the received batches from its memory pool (if they have already been added). 
Conversely, transactions that have been received from clients shortly after appearing in batches pulled from the primary do not get added to the memory pool, as they have already been processed.

\begin{figure}[h]
	\centering
	\subfloat[\centering A transaction sent from a client is sent to all router nodes and is then dispatched to the batcher node of each party]{{\includegraphics[width=5cm]{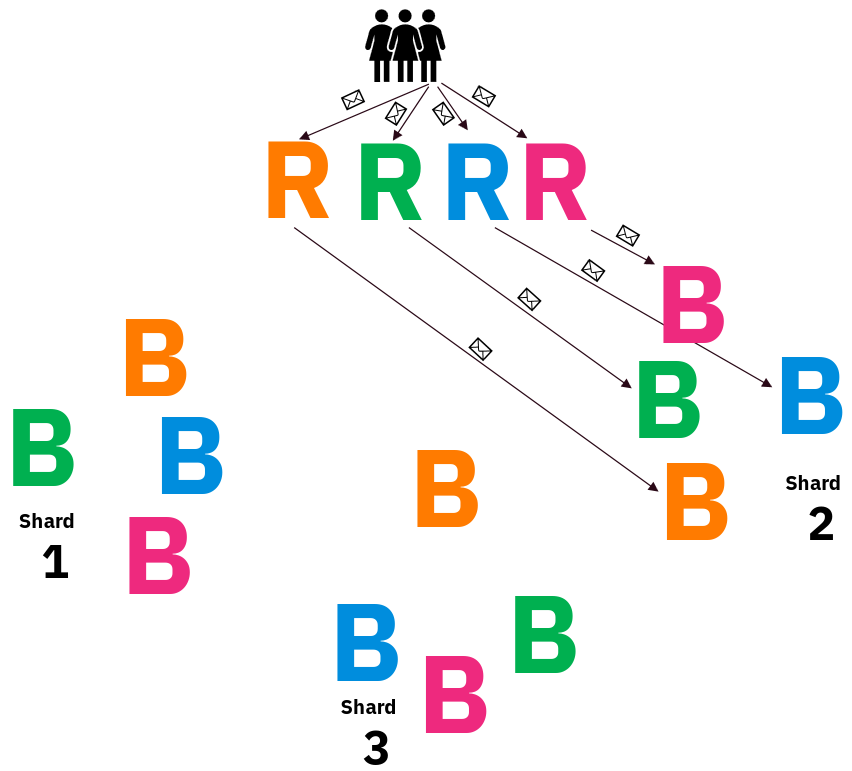} }}%
	\qquad
	\qquad
	\qquad
	\subfloat[\centering In each shard, the primary batcher node broadcasts a batch of transactions to the secondary batcher nodes]{{\includegraphics[width=5cm]{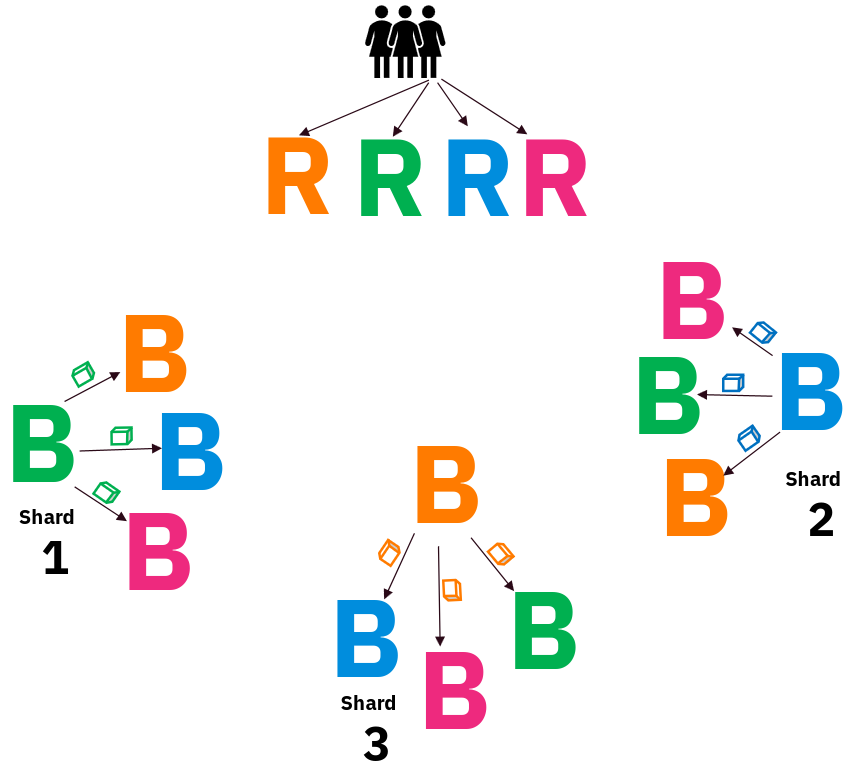} }}%
	\caption{A birds eye view of router and batcher nodes of all parties (a) routing a transaction to the appropriate shard, and (b)  broadcasting transaction batches}%
	\label{fig:routingbatching}
\end{figure}

Figure \ref{fig:routingbatching} shows the transaction flow from clients to batchers through the router nodes, and also how transaction batches are broadcast among the batchers in each shard.

Once a batch is persisted to disk, each batcher be it primary or secondary, creates a BAS by signing over a message $<sequence,digest,shard,primary>$ where $digest$ is a hash or a Merkle tree root of the transactions in the batch, and totally orders the BAS by sending it to all consensus nodes, acting as a client to the consensus nodes. 
It is the responsibility of each batcher to ensure that the BAS is sent to at least $N-F$ consensus nodes.

\begin{figure}[h]
	\fbox{
		\begin{minipage}[t]{0.5\linewidth}
			\vspace{0pt}  
					\begin{algorithm}[H] \label{alg:primbatcher}
						\KwIn{\\$~~~~~$Private key $sk$,\\
							$~~~~~$ledger $\mathcal{L}$,\\
							$~~~~~$Shard ID $S_{ID}$,\\
							$~~~~~$Total Order Broadcast $TO$,\\
							$~~~~~$Transaction memory pool $\mathcal{M}$}
						\While{I am primary}{
							b $\leftarrow \mathcal{M}.Get()$\;
							$\sigma$ $\leftarrow$ sign($sk$, b.Seq || b.Digest || $S_{ID}$ || $batcher_{ID}$) \;
							$\mathcal{L} \leftarrow \mathcal{L} ~ || ~ b$ \tcp*{Persist batch}	
							$TO.Broadcast( \big< \sigma, b.Seq, b.Digest,$\
							$S_{ID}, batcher_{ID} \big> )$\;	
						}
						\vspace{36mm}
						\caption{Primary batcher}
						\rememberlines
					\end{algorithm}
		\end{minipage}%
		\hfill\vline\hfill
		\begin{minipage}[t]{0.5\linewidth}
			\vspace{0pt}
					\begin{algorithm}[H] \label{alg:secbatcher}
						\resumenumbering
						\KwIn{\\$~~~~~$Private key $sk$,\\
							$~~~~~$ledger $\mathcal{L}$,\\
							$~~~~~$Shard ID $S_{ID}$,\\
							$~~~~~$Total Order Broadcast $TO$,\\
							$~~~~~$Transaction memory pool $\mathcal{M}$,\\
							$~~~~~$Stream of batches from primary $\mathcal{B}$}
						\While{I am not primary}{
							$seq$ $\leftarrow \mathcal{L}.Height()$\;
							$b$ $\leftarrow \mathcal{B}.RetrieveBatch(seq)$\;					
							\If{$invalidTxInBatch\left(b\right)$}{ 
								\tcp{Complain about primary}
								$\sigma$ $\leftarrow$ sign($sk$, $term_{t}$ || $S_{ID}$)\;	
								$TO.Broadcast$($\big< \sigma, term_{t}, S_{ID}\big> )$
								\Return{}	
							}
							$\mathcal{L} \leftarrow \mathcal{L} ~ || ~ b$		 \tcp*{Persist batch}		
							$\mathcal{M}.Remove(b.Requests)$\;
							$\sigma$ $\leftarrow$ sign($sk$, b.Seq || b.Digest || $S_{ID}$ || b.Primary)\;						
							$TO.Broadcast( \big< \sigma, b.Seq, b.Digest,$
							$S_{ID}, b.Primary \big> )$\;
						}
						\caption{$~~$Secondary batcher}
					\end{algorithm}
		\end{minipage}
	}
	\caption{Pseudo-code for batchers (primary and secondary). Once a batch is appended to the ledger (line 4), it will eventually be outputted by $\mathcal{B}$.}
	\label{fig:batchers}
	
	\vspace{-5mm}
	
\end{figure}

Figure \ref{fig:batchers} shows the pseudo-code for batcher nodes in both primary and secondary roles, where sending the BAS to consensus nodes is abstracted by a \emph{Total Order} primitive.

\subsubsection{Censorship Resistance in Byzantine Fault Tolerant Protocols}

A Byzantine Fault Tolerant (BFT) consensus protocol that ensures censorship resistance ensures that each transaction is eventually included in some finalized block. 
A common approach involves non leader (follower) nodes recording the time when a transaction is received. 
If a transaction is not included in a block sent by the leader within a specific time period, the follower nodes forward the transaction to the leader.
If the leader still does not include the transaction in subsequent blocks, a leader change protocol is initiated. 
Subsequently, the leader role is rotated, and one of the follower nodes becomes the new leader. 
Another technique of ensuring censorship resistance is a periodical leader rotation. 
By periodically switching the leader role, a transaction censored by a malicious leader node is eventually included in a block once a correct node becomes the leader.

\subsubsection{Censorship Resistance in Arma Batchers}
\label{sec:censorship-resistance}
Arma guarantees that if a client submits its transaction to all parties, it will eventually be included in a block without the need for the client to retry the submission.
In Arma, a similar technique to the first censorship resistance method is employed to ensure that batcher primary nodes do not censor transactions from clients. 
However, instead of initiating a leader change protocol among the batcher nodes, Arma utilizes a complaint voting mechanism.

In Arma, if a secondary batcher node receives a transaction that has not been included in a block sent by the primary batcher node for a period of time, it sends the transaction to the primary node itself. 
This is to ensure that the primary batcher received the transaction. 
Otherwise, a faulty client may have not sent the transaction to the primary batcher node and it will be falsely accused of censorship.
Only if the transaction is still not received in a batch from the primary batcher node after sending it directly, the secondary node suspects the primary node of censorship.
The secondary batcher node then sends a complaint vote to all consensus nodes. 
These complaint votes are totally ordered alongside the BASs.
Once a threshold of complaint votes is gathered against a specific primary batcher, the batcher nodes collectively designate the new primary batcher.

\subsubsection{Verifying Batches Sent from Primary Batcher Nodes}
As previously mentioned, part of the censorship resistance mechanism of batchers is having the secondary batcher nodes send transactions that were not included in a batch sent from the primary batcher.
This is done by the secondary batcher sending the transaction to the router node of the primary batcher's party. 
This guarantees that the secondary batcher nodes cannot send bogus transactions and compete with authentic ones over space in the memory pool of correct primary batchers.

What remains is preventing malicious primary batcher nodes from sending bogus transactions in their batches to the secondary nodes. While such transactions do not occupy space in the memory pool of the secondary nodes, they still should be avoided as they needlessly wastes resources and lower the effective throughput of the system when they are sent over the network and written to disk. 
The overwhelmingly common mitigation approach is to have the secondary batchers verify the transactions sent by the primary batcher. This approach indeed prevents the inclusion of bogus transactions, however, it is very resource intensive and does not scale horizontally as each secondary node must verify each transaction it receives from the primary batcher\footnote{Of course, a secondary batcher may distribute the verification to several machines and then aggregate the results, but this requires to send the batch and wait for results, which carries a latency overhead}. Arma instead employs a \emph{probabilistic approach}: Every secondary batcher node chooses and verifies a random subset of transactions in each batch received from the primary batcher node. 
More specifically, each correct secondary batcher picks $R$ transactions to verify among the $M$ transaction in a batch. $R$ is determined in so that if the batch presents more than $K$ invalid transaction, then a correct secondary batcher will detect this with high probability.


When a secondary node finds an invalid transaction, it issues a complaint to overthrow the primary and select a new one.

\subsubsection{Complaint Threshold for Rotating the Primary Batcher}
Effectively, the batcher nodes utilize the BFT consensus protocol
executed by the consensus nodes as a bulletin board that tracks, for
each shard, which party runs the primary batcher node for a given term. This approach
facilitates agreement on the misbehavior of faulty primary batcher nodes, and implicitly provides auditability on the misbehavior.

The threshold of complaint votes that induces a change in the primary batcher node is $F+1$ complaints from distinct nodes about a specific term. When $F+1$ complaint votes are collected, the term is incremented by 1 and a new batcher is then made primary. 
By defining the threshold to be $F+1$, we are assured that at least one correct party considers the primary batcher node of a shard to be faulty, and that no coalition of the $F$ faulty parties can overthrow a correct primary batcher node.

\subsection{Inducing Total Order Across Batches via BFT Consensus}
Once a batcher node persists a batch, it sends its BAS to all consensus nodes, who totally order the BASs by running a BFT consensus protocol. For simplicity, we say that the BFT protocol advances in rounds, and in each round it totally orders some BASs. The rounds do not refer to communication rounds or stages in the consensus protocol, but rather to the time elapsed between a leader sending a proposal to that proposal being committed by the consensus nodes.

\subsubsection{Collecting Batch Attestation Shares}
\label{sec:collecting-batch-attestation-shares}
Each batcher, whether primary or secondary, sends a BAS upon successfully persisting a batch to disk. 
The BASs  are totally ordered by consensus in some consensus round, and during the consensus protocol, the Arma consensus nodes determine which batch has collected enough BASs. 
Once $F+1$ distinct shares of the same BAS, i.e.,  $<sequence,digest,shard,primary>$,
are totally ordered, the associated batch can be successfully retrieved from at least one correct party. 
Therefore, the criteria for the inclusion of a batch in the total order is set to collecting $F+1$ BASs for that batch.

After $F+1$ BASs for a specific batch are totally ordered, the consensus nodes collaborate to assemble a quorum of signatures over a block header corresponding to that batch.
The signed block header contains a hash pointer to the previous block header extending thus the chain.
Each block header contains a digest of its corresponding batch, which is either a hash or a Merkle root of that batch. In this case, a batch corresponds to a block.
For efficiency purposes, a block can contain multiple batches and the header multiple batch digests, each corresponding to a batch in the block. 
This amortizes signature verification over several batches making the verification cheaper.


\begin{figure}[hpt]
	\fbox{
		\vspace{0pt}
		\begin{algorithm}[H]\label{fig:consensuss}
			\KwIn{\\$~~~~~$An ordered set of pending BASs
				from earlier rounds $P=\{bas_1, bas_2, ..., bas_m\}$,\\
				$~~~~~$A set of  BASs proposed by the  leader consensus node $B=\{bas_1, bas_2, ..., bas_k\}$}
			\KwOut{\\
				$~~~~~$An ordered set of pending BASs for future rounds $P$,\\
				$~~~~~$An ordered set of sets of distinct $F+1$  (threshold) BASs $T$
			}
			$P \leftarrow P \cup B$ \;
			$T \leftarrow \emptyset$ \; 
			$C \leftarrow \emptyset$ \; 
			{
			\tcp{Aggregate BASs according to their batches and count them}
				\ForEach{$bas\in P$}{
					$k \leftarrow < bas.Seq, bas.Shard, bas.Digest, bas.Primary \big>$ \; 
					$C[k] \leftarrow C[k]+1$ \;
					$T[k] \leftarrow T[k] \cup bas$ \;
				}			
				\ForEach{$k\in C$}{
					\If{$C[k] < F+1$}{ 
						$T \leftarrow T \setminus T[k]$ \tcp*{Find batches with at least a threshold of BASs}			
					}
				}
				
				$P \leftarrow P \setminus T$ \tcp*{Pending BASs are those without a threshold}			
				
				\KwRet{P, T}
			}	
			\caption{$~~$Pseudocode for processing batch attestation shares}
		\end{algorithm}
	}	
\end{figure}

\subsubsection{Collecting Thresholds of Batch Attestation Shares}
During the process of totally ordering BASs, a very likely scenario is one where $F+1$ BASs is totally ordered, but in different BFT rounds. 
To handle this situation, Arma maintains a list of pending BASs. 
Once $F+1$ BASs corresponding to a batch are collected, the consensus nodes  can collaborate to create a block header corresponding to that batch, sign over it and include it in the hash chain.
In Algorithm~\ref{fig:consensuss} we explain how thresholds of $F+1$ BASs for the same batch are grouped together and extracted even though they arrive in different consensus rounds. 
The BASs agreed upon in a round of BFT are input of Algorithm~\ref{fig:consensuss}, which collects thresholds of BASs from both the current consensus round and previous rounds. The leftover BASs for which a threshold was not found are then passed as input to the next consensus round.

This approach ensures that all consensus nodes will process BASs deterministically and sign over the same block headers. 
However, there is a non trivial challenge in this design: After $F+1$ distinct BASs are accumulated in the pending list, they are purged by Algorithm \ref{fig:consensuss}, but additional BASs (of the same batch) may then accumulate in future rounds. 
This is because more than $F+1$ batcher nodes may submit their BASs, but there is no guarantee that they end up being proposed in the same BFT round. 
For example, in a $4$ node setting with a single faulty node, the threshold is $F+1=2$, however, $3$ BASs are created. 
We denote a BAS that is totally ordered by consensus in a later round after collecting $F+1$ BASs, as an \emph{orphaned} BAS. 
Clearly, in the previous scenario, the orphaned BAS will never be purged from the pending list. 
It follows that as the system processes new transactions, the pending list will include more and more orphaned BASs that will never be purged.

Another problem is that accumulating too many BASs can lead to the creation of two block headers with the same batch digest.
For example, in a system with four nodes where all nodes are non-faulty, two thresholds of $F+1$ BASs may be collected.

Due to space limitations, we discuss our solution to these problems in Appendix \ref{GC}.

\begin{figure}[hpt]
	\fbox{
		\vspace{0pt}
		\begin{algorithm}[H] \label{alg:assembler}
			\SetKwBlock{DoParallel}{do in parallel}{end}
			\SetKw{Continue}{continue}
			\KwIn{\\$~~~~~$Streams of batches from batcher nodes $\{\mathcal{B}_1, ..., \mathcal{B}_k\}$, one for each of the $k$ shards.\\
				$~~~~~$A stream of block headers from consensus nodes $\mathcal{H}$.\\
				$~~~~~$An index for blobs of data $\mathcal{I}$.\\
				$~~~~~$A ledger of blocks $\mathcal{L}$.
			}	
			{
				
				\DoParallel{
					\ForEach{$shard\in \{1, .., k\}$}{
						$B \leftarrow \mathcal{B}_{shard}$ \; 
						$\mathcal{I}.Index(B)$ 
					}
					\For{true}{
						$h \leftarrow \mathcal{H}$\;
						\uIf{$\mathcal{I}.Exists(h.Digest)$}{
							$B \leftarrow \mathcal{I}.Retrieve(h.Digest)$\;
							$\mathcal{L}.Append(\big< h,  B\big>)$\; \label{alg:assembler:line:append}
						}\Else{
							\Continue	
						}
					}
				}				
			}
			\caption{$~~$Pseudo-code for assembling blocks from batch attestations and transaction batches}
		\end{algorithm}
	}	
\end{figure}

\subsection{Block Assembly}
\label{sec:assemb}
After each round of consensus in Arma, one or more block headers are persistently stored by the consensus nodes. 
These block headers include digests referencing batches, and a monotonically increasing sequence number. 
The assembler nodes retrieve the block headers from the consensus nodes, while the batches are retrieved from the batcher nodes.

Each block header carries a quorum of signatures from the consensus nodes, which the assembler nodes verify. 
This and the block sequence number guarantee that all assembler nodes commit the same headers in the same order.
For each block header, an assembler node is responsible for retrieving the corresponding batches, attaching them together to form a complete block and subsequently writing it to the ledger as depicted in Algorithm~\ref{alg:assembler}. 

By storing the entire block, which encompasses the block header, the associated transactions, and a quorum of signatures on the block header (which contains a collision-resistant digest of the transactions) in the ledger, the assembler node is effectively an archivist of the Arma system.

As the shards independently generate batches at a potentially different rate from the rate that the corresponding headers are totally ordered, it is impossible to predict the exact number of batches that are going to be fetched before their  corresponding headers are received. Consequently, it is not practical for the assembler to keep a batch in memory until its block header arrives. Thus, batches are promptly written to disk the moment they are fetched by the assembler.

\subsection{Primary Batcher Failover}

While it is easy to see how to deal with a failure of a router, a consensus node, or an assembler, a protocol for dealing with batcher failures must be presented. 
We recall that routers are stateless components, consensus nodes already run a BFT protocol that copes with failures, and assemblers do not interact with one another. On the other hand, primary and secondary batchers interact with each other and the failure of a shard's primary batcher impacts its secondary batchers. 

In BFT consensus protocols, both the view changes (changing of the leader node) and the blocks are totally ordered. 
Meaning that if block $b$ is proposed by a leader of view $v$ and later a view change occurs, incrementing the view to $v+1$ and changing the leader will have all (correct) nodes agree on whether block $b$ was committed as part of view $v$ or $v+1$.

However, in the Arma protocol, the retrieval of batches by secondary batchers from the primary is independent of the retrieval of the updates from the consensus, and specifically, from the primary batcher changes.
As a result, the progress of batches operates asynchronously with respect to the rotation of the primary batcher. 
This may introduce problems such as the loss of transactions, due to a primary batcher change before sufficient BASs reach the consensus, causing transaction batches not to be included in a block. 
If transactions are lost then by definition the protocol is not censorship resistant. Next we discuss the different scenarios and show that by following the Arma protocol transactions are never lost.

We take a look at batch $b^{p}$ sent by primary batcher $p$. Batcher $p$ later crashes and a secondary batcher $q$ takes its place as a primary.
In the simple case that $q$ received $b^{p}$, in addition to a threshold of $F+1$ batchers sending their BASs of $b^{p}$, the batch is ordered by the consensus and there is no loss of transactions.

Another simple case is when a threshold of $F+1$ batchers send their BASs of $b^p$, however the next primary batcher $q$ didn't receive batch $b^p$. This could happen for various reasons and can lead $q$ to include the same transactions of $b^p$ in a new batch made by $q$ ($b^q$). Therefore, transactions may be included twice (in $b^p$ and in $b^q$) but will not be lost. 

If $q$ receives batch $b^p$ and then it becomes the primary it needs to decide whether to propose its own batch $b^q$ with the same transactions of $b^p$. For this, $q$ inspects the BASs totally ordered by the consensus and sets $b^q$ to contain the transactions of $b^p$ if less than $F+1$ BASs were totally ordered for $b^p$.

Finally, the censorship resistance technique is employed to make sure the transaction in $b^p$ are not lost. If less than $F+1$ correct batcher nodes received $b^p$ and $q$ does not include the transactions in $b^p$ in a new batch, $q$ will be considered as malicious and be replaced as the primary batcher. Indeed, since  less than $F+1$ correct batchers received $b^p$, there is at least one transaction in $b^p$ that was sent by a correct client that is still in the memory pool of at least $F+1$ correct batchers \footnote{$2F+1$ total correct batchers, out of which $F$ or less received batch $b^p$, leaving at least $F+1$ not receiving $b^p$.} which will complain after a timeout, leading to a primary batcher failover.

\section{Integration into Hyperledger Fabric}
\label{sec:integration}

In this section, we explain how Arma can be integrated into Hyperledger Fabric.
Additionally in the next section (Section~\ref{sec:eval}), we evaluate the performance of a prototype of a Fabric \cite{HLF} ordering service node with the embedded Arma components.

As previously mentioned, the Arma system comprises different types of nodes, including routers, assemblers, batchers, and consensus nodes.
Conversely, in Fabric, there are only two types of nodes:
Peers, responsible for processing transactions, and ordering service nodes, tasked
with receiving transactions from clients and totally ordering
them into blocks for retrieval by the peers. Within this structure,
it is feasible to embed many of the Arma components within a Fabric
ordering node. While this topology may not fully leverage the scalability
capabilities of Arma, it still enables a substantial increase in the
throughput of the ordering service by an order of magnitude.

By integrating Arma into the Fabric ordering service node, Fabric can benefit from
the improved efficiency and performance offered by Arma's consensus
protocol. This integration represents a significant enhancement for
the ordering service in Fabric, allowing for a substantial
boost in transaction throughput.

In Fabric, ordering nodes perform two main tasks:

\begin{itemize}
	\item Verify transactions are well formed and properly signed by authorized clients.
	\item Totally order the transactions and bundle them into blocks signed by the orderer nodes themselves.
\end{itemize}

We decouple these two tasks by delegating the former task of verifying transactions to Arma router nodes.
The rest of the Arma nodes are embedded in the Fabric ordering service node, as seen in Figure \ref{fig:farma} .

\begin{figure}[h]
		\centering
		\includegraphics[width=14cm]{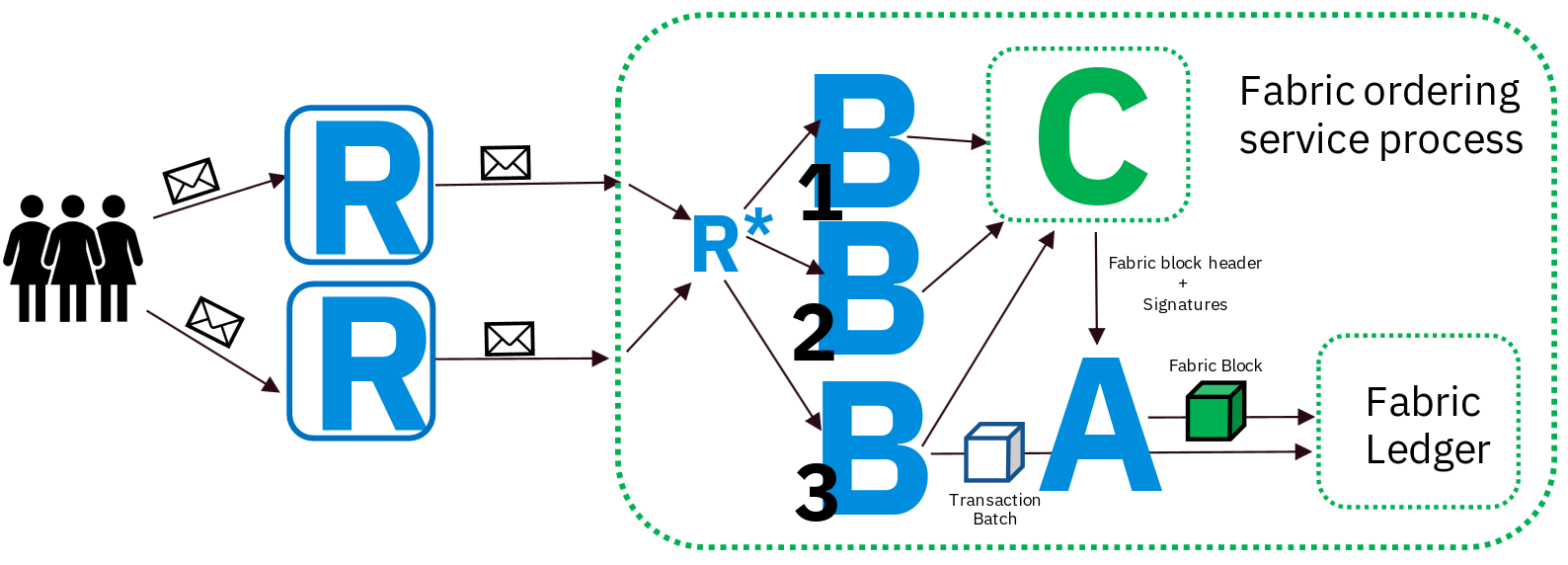}
		\caption{Arma components blue integrated alongside Fabric components (green) in a Fabric ordering node process.
			Transactions from clients are verified by router nodes and are forwarded into the ordering service node, where the $R^*$ router dispatches them to the batcher instances of appropriate shards.}%
		\label{fig:farma}
\end{figure}

As transactions arrive from the router nodes into the ordering node, they are routed to the batcher instance corresponding to the shard computed by the $R^*$ router logic. Then, transaction batches are written into the Fabric ledger as they are received from the primary batcher node. For the consensus node, we re-use Fabric's native BFT library \cite{SmartBFT} as-is, but modify its configuration to order BASs instead of Fabric transactions and to produce batches of signed Fabric block headers. More specifically, in the original Fabric implementation, the BFT orderer embeds the SmartBFT \cite{SmartBFT} consensus library. The leader broadcasts an unsigned block and the nodes sign the block header and piggyback their signatures during the agreement protocol. As a result, a quorum of Fabric signatures is collected by every node at the end of an agreement on a block. In our prototype, we made the SmartBFT leader node broadcast a batch of BASs (and complaints on primary batcher nodes), and instead of signing over the entire batch, they deterministically assemble Fabric block headers and sign over them. The resulting block headers are then passed into the assembler component. Afterwards, the transaction batches are retrieved from the Fabric ledger, and full Fabric blocks are assembled and written to the ledger, ready for retrieval by peers.

\subsection{Chaining Transaction Batches in Hyperledger Fabric}
In some distributed ledgers it is not possible to associate a block header with more than one batch. For example, in Fabric
each block header contains essential information such as a sequence
number, a hash of the previous header, and a hash of the transactions
associated with the block. The transactions are concatenated and then
hashed using a cryptographic hash function. Due to the structure of
a block header, which consists of a single hash, it is not possible
to associate a Fabric block header with more than one batch. Hence,
When integrated with Fabric, Arma's consensus nodes generate a single Fabric block header for $F+1$ BASs.
Within each set of block headers that is totally ordered via the BFT consensus protocol, the Fabric block headers are linked
together in a hash chain.

\section{Evaluation}
\label{sec:eval}

In this section we evaluate the performance of two Arma implementations. First we evaluate a prototype of Arma integrated into Fabric as described in Section~\ref{sec:integration}.
We then evaluate a prototype of a fully distributed stand-alone Arma system, deployed over LAN and WAN.

\subsection{Arma Embedded in Fabric}
\label{sec:eval-fabric}

We evaluate the performance of an Arma prototype integrated into a Fabric ordering service node (OSN). We have incorporated batchers, assemblers, and consensus nodes directly into the OSN process. The router node is simply a function that maps transactions to the appropriate batcher instance within the Fabric OSN process. Thus, each OSN process includes a router, one batcher for each shard, a consensus node, and an assembler.

\subsubsection{Experiment Setup} \label{sec:eval-fabric-setup}
The evaluation was carried out using multiple clients situated in the UK, generating transactions of two different sizes: 300 bytes in one experiment and 3.5 KB in another experiment, the latter corresponds to the standard Fabric transaction size. Each transaction was sent to all ordering service nodes in parallel. Subsequently, the clients retrieved blocks from these ordering nodes and calculated both throughput and latency.

The ordering service nodes were deployed across three distinct data centers located in the UK, Italy, and France. Latency measurements between these data centers were as follows: 10ms between England and France, 20ms between England and Italy, and 17ms between Italy and France.

For the infrastructure, the ordering service nodes were hosted on dedicated bare-metal Ubuntu 22.04 LTS machines. These machines were equipped with 96 Intel Xeon 8260 2.40GHz processors (comprising 48 cores with 2 threads per core) and 64GB of RAM. Additionally, they featured a Broadcom 9460-16i RAID 0 configuration with two SSDs.

We conducted experiments with various numbers of Fabric ordering service nodes starting from 4 nodes up until and including 16 nodes.

\subsubsection{Result Analysis}
The results of the evaluation can be seen in Figure~\ref{fig:evaluation}.
The right and left graphs show evaluation with transaction sizes of 300 bytes and 3500 bytes respectively.

\begin{figure}[h]
		\includegraphics[scale=0.13]{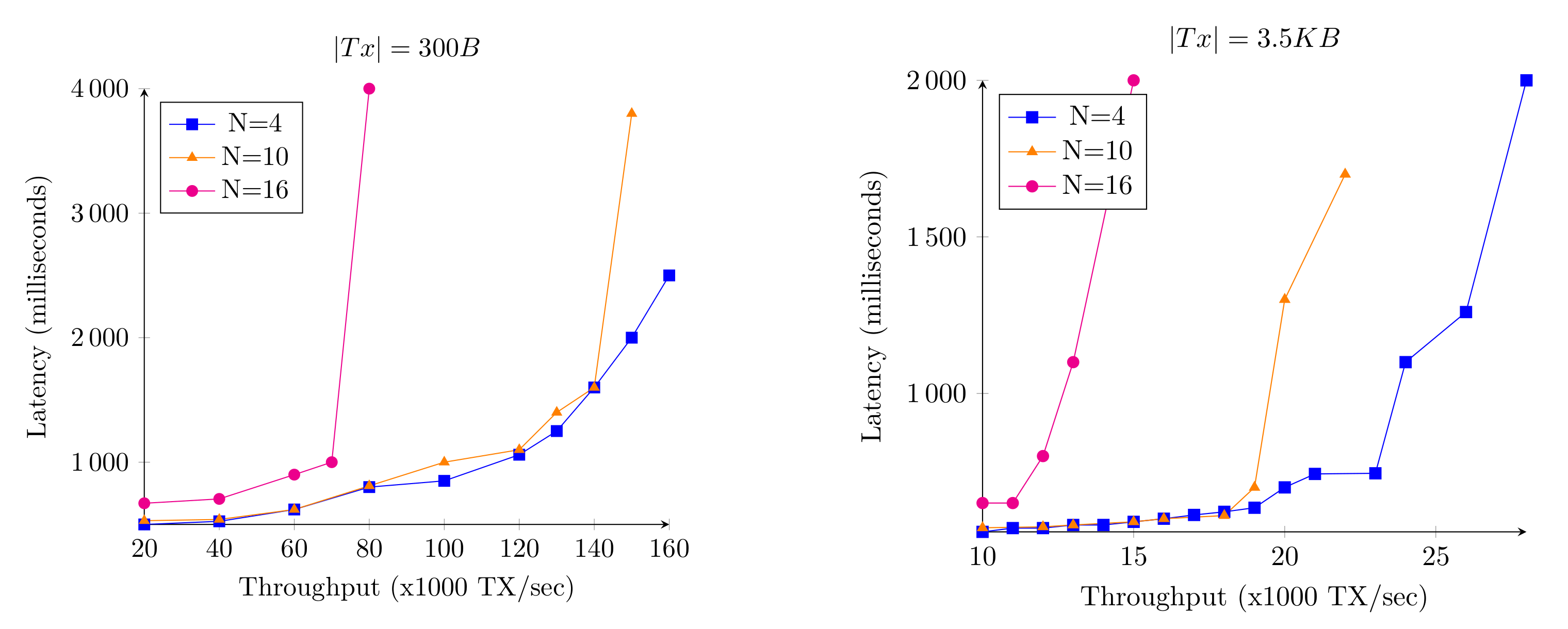}
	\caption{Performance with varying number of nodes ($N$) and varying TX size.}
	\label{fig:evaluation}						
\end{figure}

First, it can be clearly seen that the higher the amount of nodes the lower the throughput. This is expected, as more data needs to be transmitted among the nodes.

An interesting fact is while in the experiment with the 3.5KB transaction size, the transactions are $11.6$ times bigger than the experiment with the 300B transaction size, the difference in transactions per second totally ordered is only $\sim$ 6 times higher for the smaller transactions. This may imply an overhead that stems from transaction batching and censorship resistance mechanisms.

It could however be associated to the fact that parameters such as batch size and maximum batch latency in both experiments were identical, and calls for further fine tuning of parameters as a function of the transaction size.

It is worth mentioning that while our evaluation involved embedding all components within a single Fabric OSN for simplicity, as we will see in the next section, this is not the most scalable and performant deployment strategy for Arma. However, for Fabric, this represent an order of magnitude speedup in performance (approx. 2K vs. 20K TX/sec).


\subsection{Distributed Arma}

In the experiments described below we evaluated a fully distributed implementation of Arma, in which each component is running in a separate process. 

In all experiments we used 4 parties, i.e.: 4 routers, 4 consensus nodes and 4 assemblers, as well as $4\times|shards|$ batchers. A separate set of clients submitted 160B transactions to all routers. The maximal batch size was 10000, and the maximum duration of filling a batch with transactions was 500ms.

\subsubsection{LAN Deployment}

The LAN experiment had 14 bare metal servers, each with 40 CPUs and 64GB of RAM. All servers were located in the same data center, with $<1ms$ of latency between them. In all experiments a router and a consenter process were collocated on a server (4 servers), and the remaining batcher and assembler processes were evenly distributes across the other 10 servers (4+8, 4+16, 4+32, assembles+batchers, for 2,4,8 shards, respectively). Clients located in the same data center submitted transactions.

\subsubsection{WAN Deployment}

In the WAN experiment we used 44 virtual machines, each with 32 CPUs and 64GB of memory. The machines where deployed across 3 data centers in the UK, and Italy (same latencies as reported in Sec. \ref{sec:eval-fabric-setup}). Clients located in one of the data centers submitted transactions. Each process was deployed on its own virtual machine.

\subsubsection{Results}

Figure 7 depicts the results of the LAN experiments whereas Figure 8 depicts the results of the WAN experiments. 

In these experiments we let multiple clients submit TXs to the routers at a specific rate, and measure the end-to-end TX latency. As the rate is stepped up, latency increases. At some point, when the latency is $>5s$, we stop the experiment. The highest rate which resulted in latency $<5s$ is reported in the graphs.

\begin{figure}[h]
\includegraphics[scale=0.13]{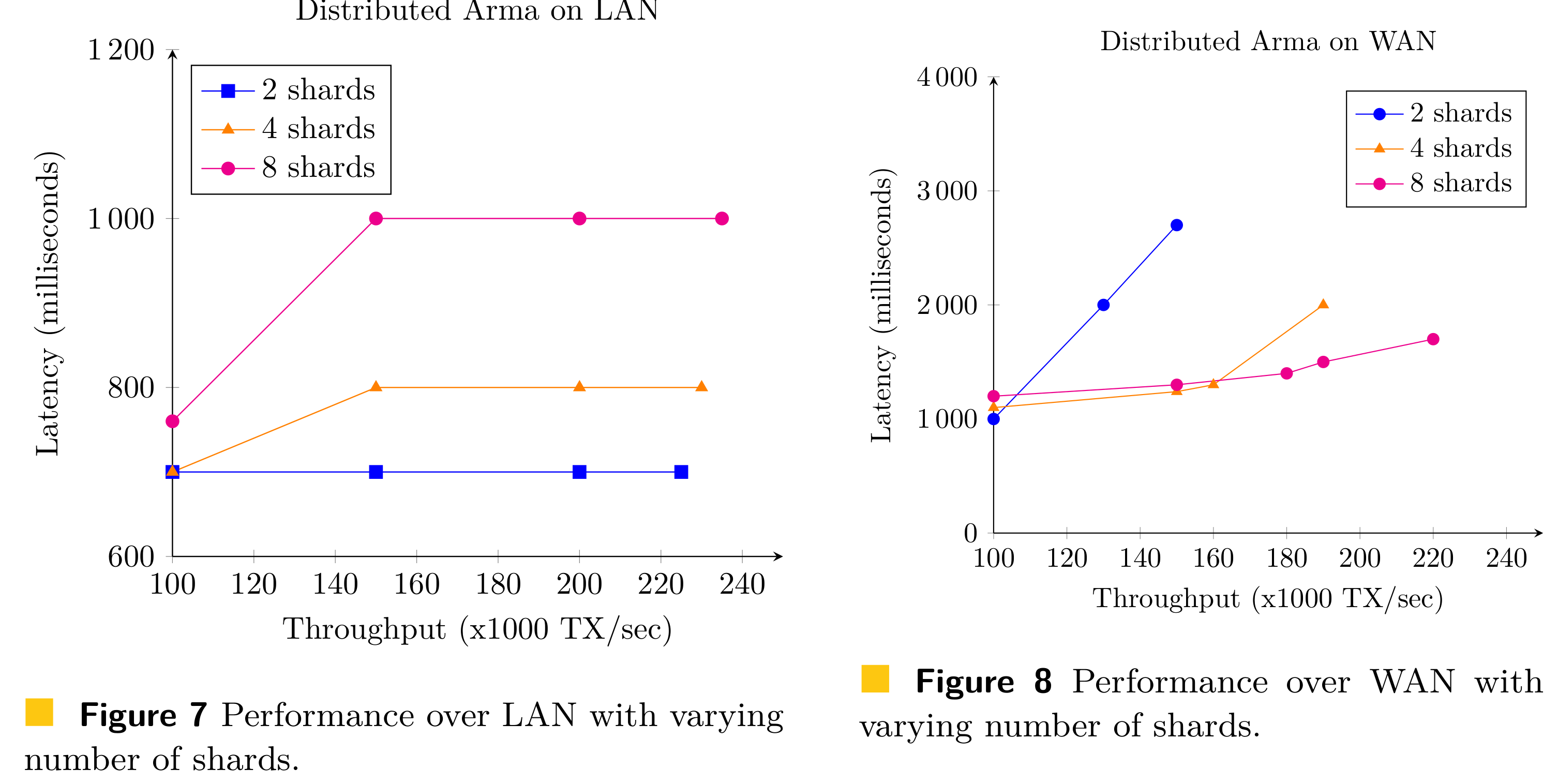}
\end{figure}

Figure 7 shows that in the LAN experiment, more shards result in better performance - from 225K to 230K to 240K TX/sec, for 2,4,8 shards, respectively. This comes at a cost of increased latency, from 700 to 800 to 1000 ms. The modest increase in performance can be attributed to the fact that even though we deployed more batcher nodes when increasing the number of shards, they were deployed on the same 10 bare metal servers.

Figure 8 shows that in the WAN experiment, more shards result in a more pronounced increase in performance -  from 150K to 190K to 220K TX/sec, for 2, 4, 8 shards, respectively. In the WAN experiment, we deployed more resources as we deploy more shards. This explains the more pronounced increase in performance.

Note that when demand is low (e.g. 100K TX/sec) deployment with different number of shards perform at about the same latency (700ms for LAN \& 1000ms for WAN). It is only when a deployment with a given number of shards become saturated and the latency starts to explode that we see the benefits of adding more shards. 

When compared to Arma embedded in Hyperledger Fabric, the distributed Arma architecture shows better performance and highlights the potential gains that can be achieved with applying the distributed Arma architecture to the Fabric ordering service.

\section{Protocol Analysis}
\label{sec:proof}
We analyze the correctness and soundness of the Arma protocol explained in Section \ref{sec:arch}. 
Arma batchers and assemblers use a BFT consensus as a primitive for totally ordering batch attestation shares and complaints. 
Therefore,  the properties of safety, liveness, and censorship resistance hold for Arma only if they hold for the underlying consensus primitive. 
We assume hereafter that up to $F<\frac{N}{3}$ parties may fail or be unreachable. 
Assuming that there is a single router node per party, this translates to having up to $F$ routers, $F$ batchers per shard, and $F$ consensus nodes fail. 
We start by formally defining the properties that the Arma consensus protocol must satisfy, after which we prove them one by one. 
Due to lack of space, the proofs are deferred to Appendix~\ref{proofs}.
For termination and censorship resistance which are both liveness properties, we assume, similarly to \cite{HotStuff}, that there exists a Global Stabilization Time (GST) after which there is an upper bound on network delay $\Delta$. 
Note that without such an assumption, it is impossible \cite{FLP} to even achieve a crash fault tolerant agreement. 
As for validity we make no such an assumption. 
We also assume it takes up to $\delta\left(\Delta\right)$ time to totally order a message through the BFT consensus used in Arma.

\begin{property}[\textbf{Termination}]\label{prop:termination} After GST, if a correct party receives a transaction $tx$ sent by a correct client, every correct assembler eventually commits a block $B$ containing $tx$.
\end{property}

\begin{property}[\textbf{Agreement}]\label{prop:agreement}If a correct Assembler node commits block $B_{i,d}$ with sequence $i$ and digest $d$ then every other correct assembler node commits block $B_{i,d}$.
\end{property}

\begin{property}[\textbf{Censorship resistance}]\label{prop:cens}After GST, there exists a time $T_{max}$ such that if a correct client submits a transaction $tx$ to at least $N-F$ parties, every correct assembler commits a block $B$ containing $tx$ within $T_{max}$ time.
\end{property}

\begin{property}[\textbf{$\alpha$-validity}]\label{prop:val}If a correct assembler commits block $B$ containing $tx$, then the probability that a client submitted $tx$ is at least $\alpha$.
\end{property}

Note that unlike protocols such as Mir-BFT \cite{Mir-BFT} where each batch sent by the leader is entirely verified by the correct nodes achieving 1-validity, Arma relaxes this property by allowing a portion of transactions sent from a primary batcher to be totally ordered without verifying them. As explained in Section~\ref{sec:design}, this method can be employed when Arma is only used for totally ordering transactions, and there is an additional transaction processing layer that consumes blocks produced by Arma batchers, and ignores invalid transactions. 
A prominent example for such a system is Hyperledger Fabric~\cite{HLF} which optimistically executes transactions on a snapshot of the world state, after which it totally orders them by a consensus service and later on validates the transactions, only taking into account transactions their read set are up to date.

\section{Conclusions}
\label{sec:future}

At its essence, Arma formulates a technique for amplifying the performance of a consensus protocol by decoupling data dissemination and validation from the actual consensus mechanism. 
As demonstrated by the evaluation of our Arma prototypes, it exhibits considerable potential in enhancing the scalability and performance of BFT consensus. 
Two key insights that stem from this work are the ability to enhance the performance of a slow, non-pipelined consensus like SmartBFT \cite{SmartBFT} by employing the methodologies explained, and how to add censorship resistance and deduplication to a system that decouples data dissemination and consensus like 'Narwhal and Tusk' \cite{Narwhale}.

As can be seen by our performance evaluation, while components run by a party can be integrated within the same machine similarly to popular protocols such as PBFT \cite{PBFT} and HotStuff \cite{HotStuff}, the throughput is smaller than the distributed Arma deployment where each component runs on its own machine. 
We recognize the need for additional evaluation of the distributed Arma deployment, namely with more parties, and consider this future work.

\newpage
\bibliography{references}

\newpage
\appendix

\section{Proofs}
\label{proofs}

Before we prove property \ref{prop:termination}, we prove a lemma:

\begin{lemma}[Primary rotation] If $F+1$ correct parties in a shard receive a transaction $tx$ but a primary batcher in that shard does not include it in a batch, then a different node becomes the primary for that shard.
	\label{lemma:primary rotation}
\end{lemma}

\begin{proof}
	As written in Section \ref{sec:batching}, if a transaction is received by $F+1$ secondary batchers, it enters their memory pool. The only way to be removed from the memory pool is via receiving a batch from the primary that includes it. 
	As already mentioned, if after a configured period of time, the transaction is not removed from a memory pool of a secondary batcher, it sends a complaint to the BFT consensus nodes. Once complaints from $F+1$ parties have been totally ordered through the BFT consensus, the next batcher in the shard is designated as primary, in a deterministic manner (e.g., using round-robin).
\end{proof}

\begin{theorem}[Termination] Arma satisfies termination (property \ref{prop:termination}). \label{theorem:termination}
\end{theorem}
\begin{proof} 
	A correct party receives a transaction via its router and then forwards it to its batcher. Until the client gets an acknowledgement of the transaction being enqueued in the batcher, it does not consider the transaction to be successfully submitted to that party. A correct client sends a transaction to the routers of all parties, but the delivery of the transaction to up to $F$ parties can fail due to faulty behavior. 
    As mentioned in Section \ref{sec:batching}, a correct client considers a transaction as successfully submitted once it has been received by $N-F$ routers. In this manner, the transaction will reach at least $F+1$ correct batchers.
	A malicious primary batcher may not include the transaction in a batch. In such a case and since at least $F+1$ correct batchers received the transaction, Lemma \ref{lemma:primary rotation} guarantees that there exists a period of time after which the faulty primary batcher ceases being the primary for the shard.

	It remains to argue that a correct primary batcher will not be replaced due to complaints from malicious batcher nodes.
	Indeed, the only way a correct primary batcher can be designated to not be the primary is if $F+1$ complaint votes against its party are collected. 
	It follows that $F$ malicious nodes attempting to vote out a correct primary need a correct node to also send a complaint vote against that primary. 
	A correct secondary sends a complaint vote against a primary in one of two cases: The primary does not include in a batch a transaction that the secondary has received and dispatched to the primary or the primary includes invalid transactions. 
	Batches generated by a correct primary will never satisfy these cases, and therefore, a correct secondary will never cast a complaint vote against it. A correct primary, therefore, will never be falsely ejected by the faulty nodes. 

	We next show that the transaction will be included in a future block by all correct assemblers.
	A correct primary batcher that receives $tx$ inserts it into its memory pool $\mathcal{M}$. Then, as shown in Algorithm $\ref{alg:primbatcher}$, a secondary batcher will eventually fetch a batch $b$ from the primary which will contain $tx$ and append the batch $b$ to its ledger. Eventually, $b$ will be replicated to at least $F$ correct nodes via lines 3-4 in Algorithm \ref{alg:secbatcher}, and the corresponding batch attestation shares will be submitted to be totally ordered with the BFT consensus via line 5 in Algorithm \ref{alg:primbatcher} and line 12 in Algorithm \ref{alg:secbatcher}. 
	Then, following Algorithm \ref{fig:consensuss}, $F+1$ batch attestation shares will be collected to trigger the creation of a new block header which contains $b$'s digest, which will be signed by the consensus nodes as explained in Section \ref{sec:collecting-batch-attestation-shares}. For every block header there are $F+1$ batcher nodes from different parties that persisted the corresponding batch, therefore at least one correct node has this batch. Both block header and the batch $b$ are then eventually retrieved by correct assembler nodes as per lines 7 and 9 respectively in Algorithm \ref{alg:assembler}, and the batch $b$ and its block header are then both committed into the ledger in line 10. 
\end{proof}

\begin{theorem}[Agreement] 
	Arma satisfies Agreement (property \ref{prop:agreement}).
\end{theorem}
\begin{proof}
	Assume by contradiction that two distinct blocks $B_{i,d}$ and $B'_{i,d'}$ with sequence $i$ were committed by different correct assemblers. 
	We distinguish between two cases: one where digests $d$ does not equal $d'$ and one where the digests are identical. 
	If $d \neq d'$, then according to line~\ref{alg:assembler:line:append} in Algorithm \ref{alg:assembler}, two different headers were received from the consensus nodes. As mentioned in Section \ref{sec:assemb}, the consensus nodes assemble a quorum of signatures over each header and a correct assembler verifies these signatures. Since two quorums intersect at at least one correct node, this entails that a correct consensus node has signed two conflicting headers, contradicting it being a correct consensus node.
	Therefore $d = d'$.
	As depicted in lines 9-10 in Algorithm \ref{alg:assembler}, blocks are created by retrieving the batches matching the digests in the totally ordered headers. 
	If $B_{i, d} \neq B'_{i, d}$, then this implies that two different blocks map to the same digest, breaking thus the collision resistance property of the hash function used to compute the digest. 
	Consequently, we conclude that correct assemblers will commit to the same blocks in the same order.  
\end{proof}

\begin{theorem}[Censorship resistance] 
	Arma satisfies the censorship resistance property \ref{prop:cens}.
\end{theorem}

\begin{proof}
	A correct client submits a transaction $tx$ to all parties. It considers a transaction as submitted once $N-F$ parties confirmed the transaction has been enqueued into the batcher nodes of the shard corresponding to the transaction. If the primary batcher node of the shard is correct, then from Theorem  \ref{theorem:termination} censorship resistance holds. Otherwise, the primary batcher is faulty and therefore it does not include $tx$ in a batch. Since $tx$ was received by at least $F+1$ correct secondary batcher nodes, by Lemma \ref{lemma:primary rotation}, after a timeout period denoted $T_{censor}$, a different batcher will be the primary for the shard. The maximum number of successive faulty primary batcher nodes is $F$. Thus, after at most $F \cdot T_{censor}$ time, the primary batcher node will be correct, and according to Theorem \ref{theorem:termination} $tx$ will be included in a block $B$ that each assembler will commit. Now the time that takes to commit a transaction submitted by a correct client is upper-bounded by $T_{max}=F \cdot T_{censor}+ \delta\left(\Delta\right) + 2\cdot \Delta$. \footnote{$\Delta$ to pull the batch and headers by the assemblers, $\Delta$ for the batch to be replicated by the batchers, and $\delta\left(\Delta\right)$ to be ordered by BFT consensus.} 
\end{proof}

\begin{theorem}[$\alpha$-validity] Arma satisfies $\alpha$-validity (property \ref{prop:val}).
\end{theorem}

\begin{proof}
	For simplicity of analysis, we assume that all $F$ faulty nodes are controlled by a single adversary, and that the adversary has control over the primary batcher of a shard and $F-1$ secondary batchers.
	As the adversary already controls $F$ batchers, in order for a batch to be output by Arma, the adversary needs to collect an additional batch attestation share. 
	Therefore, it is sufficient for the adversary to send the batch to an additional correct batcher and hope it would not detect the invalid transactions within.
	The value $\alpha$ denotes the ratio of valid transactions in the batch.
	A correct secondary batcher rejects the batch if it chooses to verify one of the remaining invalid transactions in the batch. In practice, a batcher samples the transactions to verify without replacement. However, if we model the sampling as one with replacement, we get a simpler model to analyze to find the number of transactions a correct batcher must verify to meet $\alpha$-validity with high probability. 
	Let $K$ denote the number of transactions the correct secondary batcher chooses to verify in the batch.  
	Let $\mathfrak{p}$ be a system parameter that corresponds to the acceptable probability of not rejecting a batch where  $1 -\alpha$ of the transactions are invalid.
	The correct secondary batcher fails to detect an incorrect primary batcher if it samples $K$ transactions among the ones correctly submitted by the clients, as opposed to invalid transactions injected by the primary batcher.   
	This event occurs with probability $\alpha^K$.
	$K$ is accordingly determined by equation $\alpha^k \leq \mathfrak{p}$, i.e.,  $K \geq ln(\mathfrak{p})/ln(\alpha)$. 
	For a one in a billion ($\mathfrak{p}=2^{-30}$) chance of not rejecting a batch where $1-\alpha$ of the transactions are invalid, we get $K \geq \frac{-30}{log_2(\alpha)}$.  For a batch with 50\%, 25\%, 5\% invalid transactions, we get $K=30$, $K\approx72$, $K\approx405$ transactions to verify respectively.
\end{proof}

\section{Efficient Transaction Bundling}
Given that a batcher node can function as either a primary or secondary, it is necessary to prioritize efficiency in transaction bundling. 
The challenge lies in swiftly retrieving transactions from the memory pool while maintaining the order of their arrival and simultaneously allowing uninterrupted insertion of new transactions into the pool.

The Arma memory pool utilizes a mechanism where the retrieval of batches
from the memory pool has a time complexity of O(1). 
This is achieved by having transaction insertions and batch retrievals not be conflicting with each other.
At any given time, in the primary batcher's memory pool, there is a pending batch being filled and a queue of full batches awaiting dispatch. 
The batches are dispatched in the order they appear in the queue and they are deleted afterwards.
If no full batch is present in the queue, the primary retrieves the pending batch that is currently being filled. 
As transactions enter the memory pool, they fill the pending batch, and once that batch reaches a certain size, it is atomically enqueued into the queue of full batches, while a new empty batch is created instead of the now full pending batch.

\vspace{-4mm}

\begin{figure}[h]
	\centering{{\includegraphics[width=10cm]{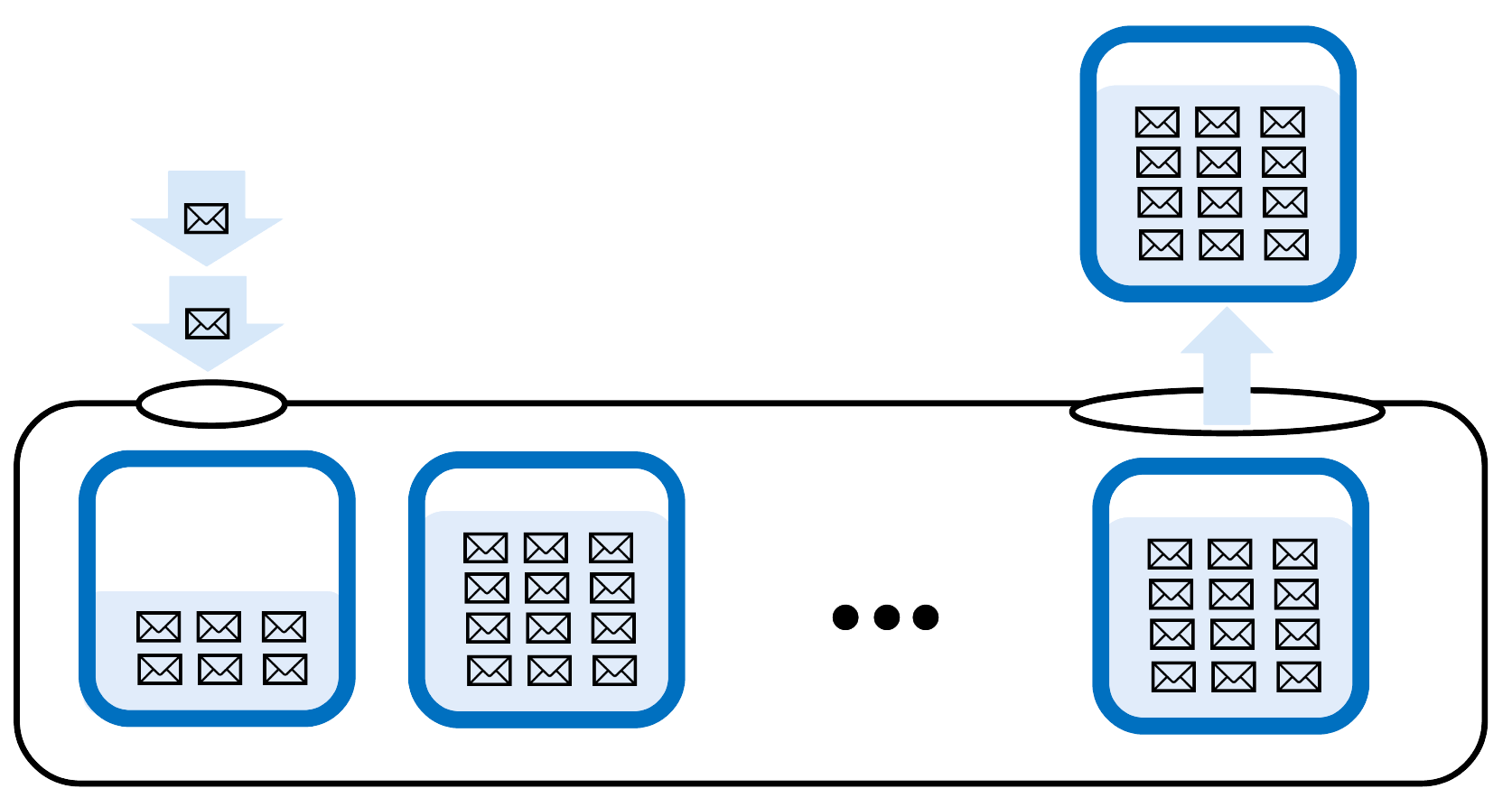} }}%
	\caption{Transaction memory pool for primary batchers}%
	\label{fig:bundling}
\end{figure}

Figure \ref{fig:bundling} depicts how transactions enter the memory pool (left side) and fill batches until they are full to form a queue. 
The next batch to be proposed is the oldest one in the queue (right side).

\section{Tracking Transactions by Secondary Batcher Nodes}
While the primary batcher's role is to quickly bundle a set of transactions,
enabling the secondary batchers to retrieve them efficiently, the secondary batchers are responsible for tracking the transactions in their memory pool and promptly detect if any transactions are not being dispatched within the expected time frame. 

In the secondary batcher's memory pool, incoming transactions are inserted into specific buckets. 
These buckets are assigned timestamps periodically and are
subsequently sealed, preventing further transaction insertions. 
When transactions arrive from the primary batcher node, they are removed
from the corresponding buckets where they were initially inserted.
A sealed bucket that remains non-empty for an excessively long duration
indicates either the transactions within it did not reach the primary
node or that the primary node is censoring the transactions. 

Efficient detection of transaction censorship involves recording the
entry time of transactions into the secondary batcher node's memory
pool and identifying instances where transactions remain in the pool
for an extended period. 
The key concept in achieving efficient censorship detection is that precision is not crucial in the case of actual censorship.
Which is why timestamps can be tracked per bucket (multiple transactions) and not per transaction.

Garbage collection occurs for sealed and emptied buckets, while sealed yet
non-empty buckets serve as indicators of potential censorship or transmission
issues. 
Figure \ref{fig:tracking} depicts how transactions are tracked in a secondary batcher's memory pool.

\begin{figure}[h]
	\centering{{\includegraphics[width=10cm]{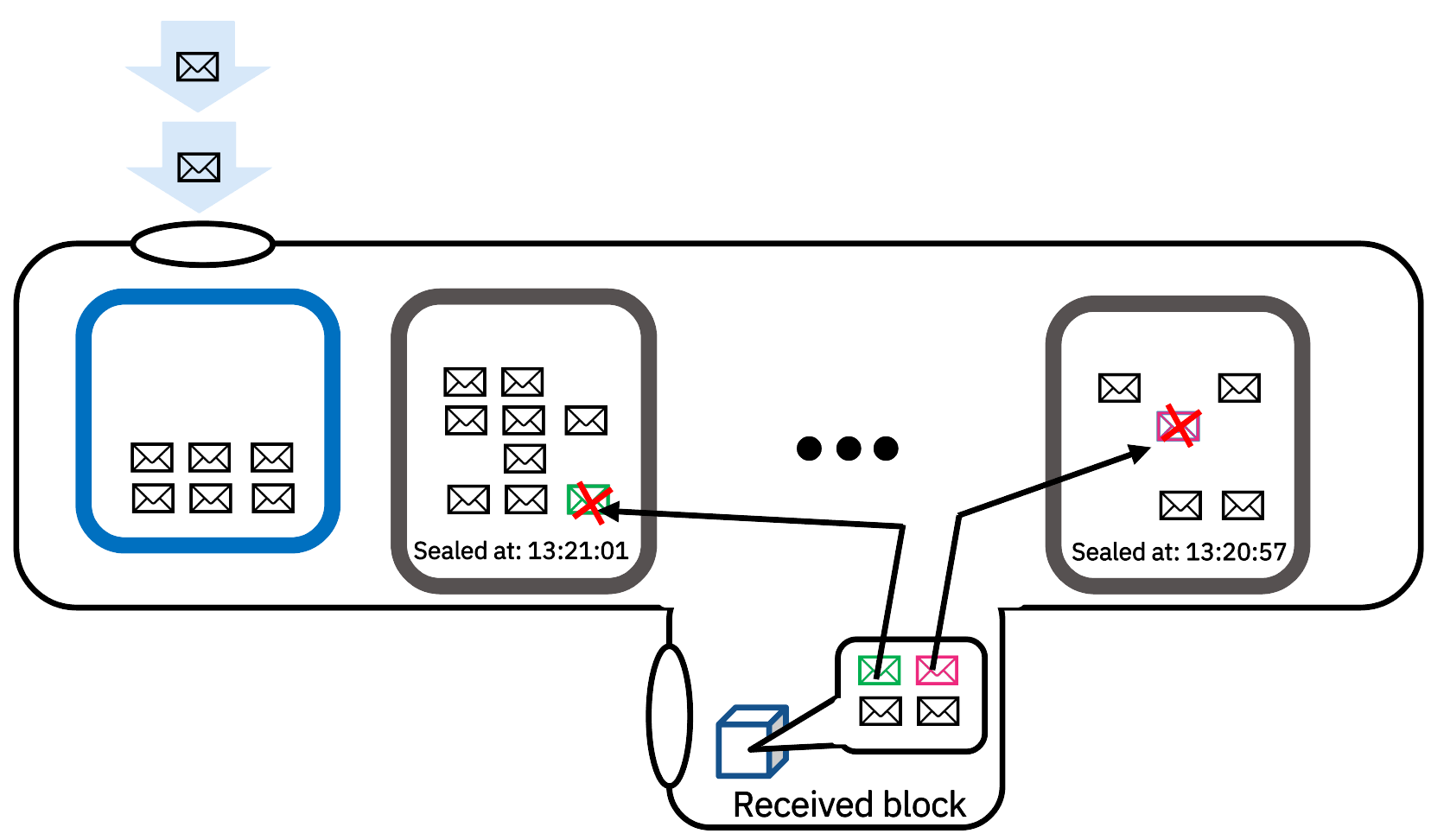} }}%
	\caption{Transaction memory pool for secondary batchers}%
	\label{fig:tracking}
\end{figure}

\section{Garbage Collection and De-Duplication of Batch Attestation Shares}
\label{GC}
To limit the consensus node's pending list from growing indefinitely and to avoid the creation of two block headers corresponding to the same batch, the consensus nodes hold database of digests of batches for which the number of corresponding BASs exceeded the $F+1$ threshold to create block headers. 

In other words, every time $F+1$ BASs are collected, the corresponding digest is added to the database. 
Then, consensus nodes know for which digests they should not attempt to create a second block header.

There are three problems to address with this design: (1) How to garbage collect the database to prevent it from growing indefinitely; (2) How to prevent malicious nodes from re-submitting old BASs after they have been garbage collected from the database, thus forcing creation of prior block headers in the hash chain; and lastly (3) how to uniformly and deterministically purge orphaned BASs from the pending list. 

To that end, Arma divides the time axis into discrete sections of equal length called epochs, which are represented by a monotonically increasing number. 

\begin{figure}[h]
	\centering{{\includegraphics[width=14cm]{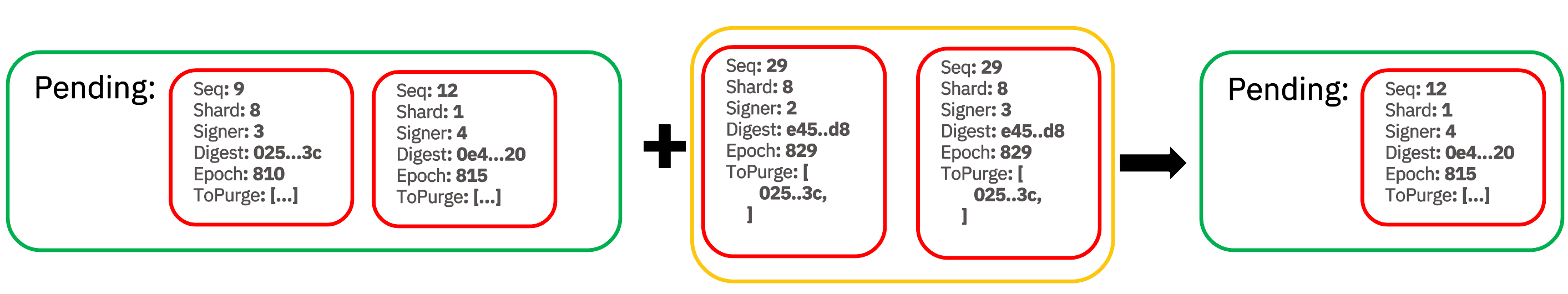} }}%
	\caption{Purging orphaned BAS of transaction batch with digest 025...3c from the pending list}%
	\label{fig:bas2}
\end{figure}

When creating a BAS, a batcher computes the current epoch and includes it in the payload of the BAS sent to the consensus nodes.
Consensus nodes prevent totally ordering BASs with epochs too far in the past, therefore the second problem of re-submission of old BASs is avoided.
Arma addresses the first and third problems without relying on an assumption that the time in all correct consensus nodes is synchronized. Instead, old BASs are pruned via vote counting: Each BAS of a shard contains references to orphaned BASs of the same shard as depicted in Figure \ref{fig:bas2} Once a BAS in the pending list has $F+1$ or more votes that point to it, it is pruned from the pending list. This ensures that at least one correct batcher node considers this BAS too far in the past and also orphaned. A BAS can only point to a BAS of an earlier sequence number or of an earlier term, so cycles are avoided.

\end{document}